%% file: HnR-arxiv.tex
\documentclass{article}

\usepackage[round,comma]{natbib}
\usepackage{fullpage}

\usepackage[utf8]{inputenc} 
\usepackage[T1]{fontenc}    
\usepackage{url}            
\usepackage{booktabs}       
\usepackage{amsfonts}       
\usepackage{nicefrac}       
\usepackage{microtype}      
\usepackage[disable]{todonotes}
\usepackage{wrapfig}

\usepackage{subcaption}

\usepackage{algorithm,algorithmic}
\usepackage{verbatim}

\usepackage{amsmath,amsthm,amssymb,amsfonts}

\usepackage{wrapfig}

\usepackage{subcaption}

\usepackage{algorithm,algorithmic}

\newcommand{\HNR}{\textsc{Hit-and-Run}}

\input{Commands}

\title{Hit-and-Run for Sampling and Planning in Non-Convex Spaces}

\author{
Yasin Abbasi-Yadkori\\
Queensland University of Technology\\
\texttt{yasin.abbasiyadkori@qut.edu.au} \\
\and
Peter L. Bartlett \\
UC Berkeley and QUT \\
\texttt{bartlett@eecs.berkeley.edu} \\
\and
Victor Gabillon\\
Queensland University of Technology\\
\texttt{victor.gabillon@qut.edu.au} \\
\and
Alan Malek \\
UC Berkeley \\
\texttt{malek@eecs.berkeley.edu} \\
}

\begin{document}

\maketitle

\begin{abstract}
We propose the \HNR~algorithm for planning and sampling problems in
non-convex spaces. For sampling, we show the first analysis of the
\HNR~algorithm in non-convex spaces and show that it mixes fast as
long as certain smoothness conditions are satisfied. In particular,
our analysis reveals an intriguing connection between fast mixing and
the existence of smooth measure-preserving mappings from a convex
space to the non-convex space. For planning, we show advantages of
\HNR~compared to state-of-the-art planning methods such as Rapidly-Exploring Random Trees. 
\end{abstract}

\input{introduction}

\section{Analysis}
\label{sec:analysis}

This section proves Theorem~\ref{thm:sampling:main}. We begin by stating  a number of useful geometrical results, which allow us to 
prove the two main components: an isoperimetric inequality in Section~\ref{sec:isoperimetric} and a total variation inequality in Section~\ref{sec:total_variation}. We then combine everything in Section~\ref{sec:everything}.

\input{conductance}

\input{isoperimetric}

\input{total-variation}

\input{final}

\input{experiments}

\section{Conclusions and Future Work}
This paper has two main contributions. First, we use a
measure-preserving bilipschitz map to extend the analysis of the \HNR~random walk to non-convex sets. Mixing time bounds for non-convex sets open up many applications, for example non-convex optimization via simulated annealing and similar methods. The second contribution of this paper has been to study one such application: the planning problem.

In contrast to RRT, using \HNR~for planning has stronger guarantees on
the number of samples needed and faster convergence in some cases. It
also avoids the need for a sampling oracle for $\Sigma$, since it
combines the search with an approximate sampling oracle. One drawback
is that the sample paths for \HNR~have no pruning and are therefore
 longer than the RRT paths. Hybrid approaches that yield short
paths but also explore quickly are a promising future direction.

\newpage
\bibliography{biblio}

\newpage
\onecolumn
\appendix

\input{proofs}

\end{document}

%% file: Commands.tex
\allowdisplaybreaks

\newcommand{\todoy}[2][]{\todo[color=red!20!white,#1]{Yasin: #2}}
\newcommand{\todoa}[2][]{\todo[color=blue!20!white,#1]{Alan: #2}}
\newcommand{\todov}[2][]{\todo[color=blue!20!white,#1]{Victor: #2}}

\usepackage{algorithm}
\usepackage{algorithmic}

\usepackage{nicefrac}

\newtheorem{thm}{Theorem}

\newtheorem{lem}[thm]{Lemma}

\newtheorem{defn}{Definition}

\newtheorem{ass}[thm]{Assumption}

\newcommand{\vol}{\textsc{vol}} 
\newcommand{\view}{\textsc{view}}

\newcommand{\abs}[1]{\left\vert#1\right\vert}

\newcommand{\Real}{\mathbb R}                        

\newcommand{\Prob}[1]{{\mathbb P}\left(#1\right)}    
\newcommand{\one}[1]{\mathbf 1 \left\{#1\right\}}           

\newcommand{\ra}{\rightarrow}

\newcommand{\argmin}{\mathop{\rm argmin}}

\newcommand{\eqdef}{\stackrel{\mbox{\rm\tiny def}}{=}}

\newcommand{\beq}{\begin{equation}}
\newcommand{\eeq}{\end{equation}}
\newcommand{\beqa}{\begin{eqnarray}}
\newcommand{\eeqa}{\end{eqnarray}}
\newcommand{\beqan}{\begin{eqnarray*}}
\newcommand{\eeqan}{\end{eqnarray*}}
\newcommand{\ben}{\begin{eqnarray*}}
\newcommand{\een}{\end{eqnarray*}}
\newcommand{\bea}{\begin{align*}}
\newcommand{\eea}{\end{align*}}

\newcommand{\cC}{{\cal C}}

\newcommand{\cR}{{\cal R}}

\newcommand{\cG}{{\cal G}}

\newcommand{\cH}{{\cal H}}
\renewcommand{\phi}{\varphi}

\newcommand{\bE}{\mathbf{E}}

%% file: introduction.tex
\section{Introduction}
\label{sec:introduction}

Rapidly-Exploring Random Trees
(RRT)~\citep{LaValle-1998,LaValle-Kuffner-2001} is one of the most
popular planning algorithms, especially when the search space is
high-dimensional and finding the optimal path is computationally
expensive. RRT performs well on many problems where classical dynamic
programming based algorithms, such as A*, perform poorly. RRT is
essentially an exploration algorithm, and in the most basic
implementation, the algorithm even ignores the goal information, which
seems to be a major reason for its success. Planning problems,
especially those in robotics, often feature narrow pathways connecting
large explorable regions; combined with high dimensionality, this
means that finding the optimal path is usually intractable. However,
RRT often provides a \textit{feasible} path quickly.

Although many attempts have been made to improve the basic
algorithm~\citep{AbbasiYadkori-Modayil-Szepesvari-2010,Karaman-Frazzoli-2010,Karaman-Frazzoli-2011},
RRT has proven difficult to improve upon. In fact, given extra
computation, repeatedly running RRT often produces competitive
solutions. In this paper, we show that a simple alternative greatly improves upon 
RRT. We propose using the \HNR~algorithm for
feasible path search. Arguably simpler than RRT, the \HNR~is a rapidly mixing
MCMC sampling algorithm for producing a point uniformly at random from a 
convex space~\citep{Smith-1984}. \todov{not that clear, you mean hnr is 
simpler that RRT? insist maybe more on what still need to be improved 
from RRT or say that you will detail later} Not only \HNR~finds a feasible 
path faster than RRT, it is also more robust with respect to the geometry of the space.

Before giving more details, we define the {\em planning} and {\em sampling}
problems that we consider.
Let $\Sigma$ be a bounded connected subset of $\Real^n$. For points
$a,b\in\Sigma$, we use $[a,b]$ to denote their (one-dimensional)
convex hull. Given a starting point $a_1$ and a goal region
$\cG\subset\Sigma$, the {\em planning problem} is to find a sequence
of points $\{a_1,a_2,\dots,a_\tau\}$ for $\tau\ge1$ such that all
points are in $\Sigma$, $a_\tau$ is in $\cG$, and for
$t=2,\dots,\tau$, $[a_{t-1},a_t]\subset\Sigma$.

The {\em sampling problem} is to generate points uniformly at random
from $\Sigma$. Sampling is often difficult, but Markov Chain Monte
Carlo (MCMC) algorithms have seen empirical and theoretical success
\citep{lovasz2007geometry}. MCMC algorithms, such as \HNR~and \textsc{Ball-Walk} \citep{VEMPALA-2005}, sample a Markov Chain on $\Sigma$ that has a stationary distribution equal to the uniform distribution on $\Sigma$; then, if we run the Markov Chain long enough, the marginal distribution of the sample is guaranteed to come from a distribution exponentially close to the target distribution. Solving the sampling problem yields a solution to the planning problem; one can generate samples and terminate when $a_t$ hits $\cG$. 
\begin{figure}
	\hspace{1cm}\includegraphics[scale=.7,natwidth=610,natheight=642]{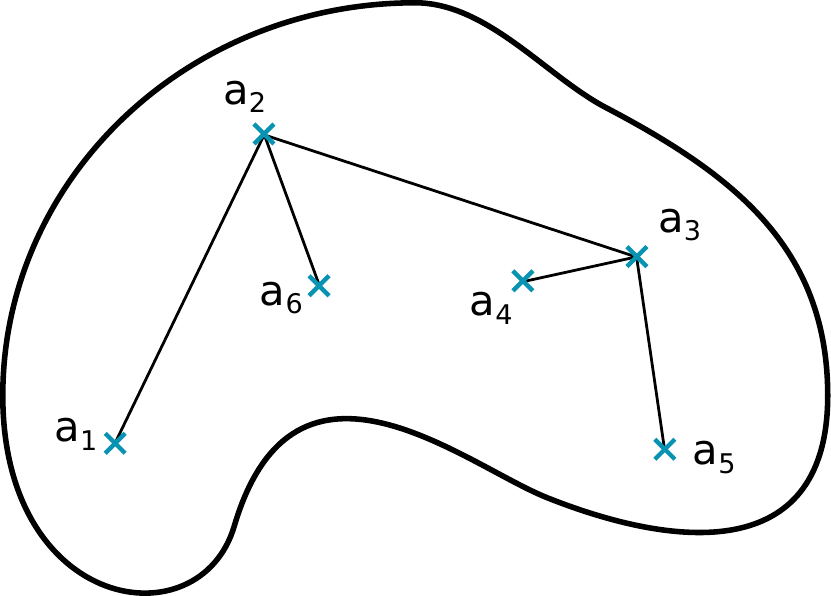}\hspace{2cm}
	\includegraphics[scale=.7,natwidth=610,natheight=642]{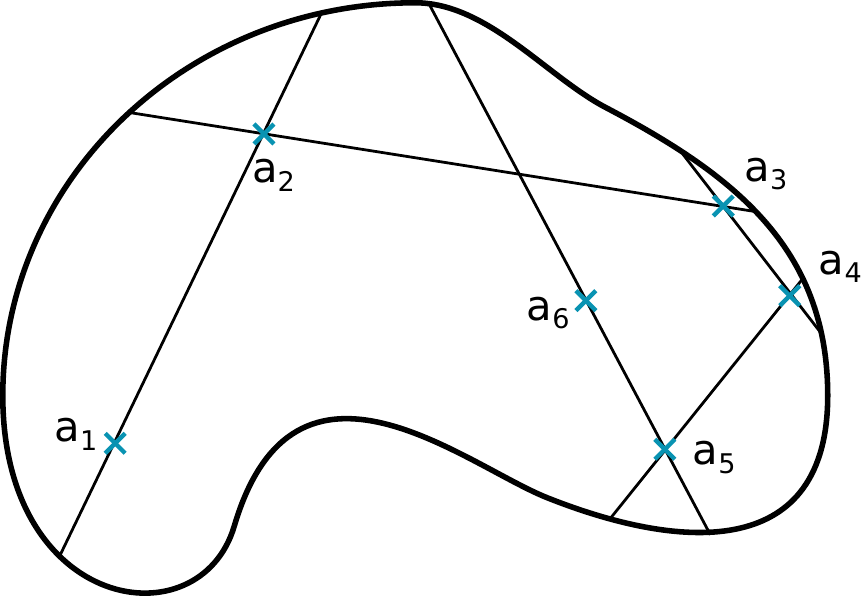}
	\caption{RRT (left) and \HNR~(right)}
	\label{fig:hitandrrt}
\end{figure}

Let us define \HNR~and the RRT algorithms (see also Figure~\ref{fig:hitandrrt} for an illustration). 
 \HNR~defines a Markov chain on $\Sigma$ where the transition dynamics are as follows. A direction is chosen uniformly at random, and $a_{t+1}$ is chosen uniformly from the largest chord contained in $\Sigma$ in this direction passing through $a_t$. This Markov Chain has a uniform stationary distribution on $\Sigma$ \citep{Smith-1984}. As a planning algorithm, this chain continues until it hits the goal region. Let $\tau$ be the stopping time. The solution path is $\{a_1,a_2,\dots,a_\tau\}$. 

On the other hand, the RRT algorithm iteratively builds a tree $T$ with $a_1$ as a root and nodes labeled as $a^n\in\Sigma$ and edges $\{a^m,a^n\}$ that satisfy $[a^m,a^n]\subseteq\Sigma$. To add a point to the tree, $a^r$ is uniformly sampled from $\Sigma$ and its nearest neighbor $a^n\in T$ is computed. If $[a^n,a^r]\subset\Sigma$, then node $a^r$ and edge $[a^n,a^r]$ are added to $T$. Otherwise, we search for the point $a^e\in[a^n,a^r]$ farthest from $a^n$ such that $[a^n,a^e]\subseteq \Sigma$. Then $a^e$ and $[a^n,a^e]$ are added to the tree. This process is continued until we add an edge terminating in $\cG$ and the sequence of points on that branch is returned as the solution path. In the presence of dynamic constraints, a different version of RRT that makes only small local steps is used. These versions will be discussed in the experiments section.   

There are two main contributions on this paper. First, we analyze the
\HNR~algorithm in a non-convex space and show that the mixing time is
polynomial in dimensionality as long as certain smoothness conditions
are satisfied. The mixing time of \HNR~for convex spaces is known to
be polynomial~\citep{Lovasz-1999}. However, to accommodate planning
problems, we focus on non-convex spaces. Our analysis reveals an
intriguing connection between fast mixing and the existence of smooth
measure-preserving mappings. The only existing analysis of random walk
algorithms in non-convex spaces is due to
\cite{Chandrasekaran-Dadush-Vempala-2010} who analyzed
\textsc{Ball-Walk} in star-shaped bodies.\footnote{We say $S$ is
star-shaped if the kernel of $S$, define by $K_S = \{x\in S : \forall
y\in S\,\, [x,y]\subset S \}$, is nonempty.} Second, we propose \HNR~for 
planning problems as an alternative to RRT and show that it 
finds a feasible path quickly. From the mixing rate, we obtain 
a bound on the expected length of the solution path in the planning problem. 
Such performance guarantees are not available for RRT.

The current proof techniques in the analysis of \HNR~heavily rely on the convexity of the space. It turns out that non-convexity is specially troubling when points are close to the boundary. We overcome these difficulties as follows. First, \citet{Lovasz-Vempala-2006-corner} show a tight isoperimetic inequality in terms of average distances instead of minimum distances. This enables us to ignore points that are sufficiently close to the boundary. Next we show that as long as points are sufficiently far from the boundary, the cross-ratio distances in the convex and non-convex spaces are closely related. Finally we show that, given a curvature assumption, if two points are close geometrically and are sufficiently far from the boundary, then their proposal distributions must be close as well. \todov{what is a proposal distribution?}
\todov{theoretical results for RRT in the litterature?}

\HNR~has a number of advantages compared to RRT; it does not require
random points sampled from the space (which is itself a hard problem),
and it is guaranteed to reach the goal region with high probability in
a polynomial number of rounds. In contrast, there are cases where RRT
growth can be very slow (see the experiments sections for a
discussion). Moreover, \HNR~provides \textit{safer} solutions, as its
paths are more likely to stay away from the boundary. In contrast, a
common issue with RRT solutions is that they tend to be close to
the boundary. Because of this, further post-processing steps are needed to
smooth the path.

\todov{insist more on the fact that you do not change hitand run or propose a new version just show that the original algo is the solution and you provide the analysis which is cool}


\subsection{Notation}

For a set $K$, we will denote the $n$-dimensional volume by $\vol(K)$, the $(n-1)$-dimensional surface volume by $S_K = \vol_{n-1}(\partial K)$, and the boundary by $\partial K$. The diameter of $K$ is $D_K = \max_{x,x'\in K} \abs{x-x'}$, where
$\abs{\cdot}$ will be used for absolute value and Euclidean norm, and the distance between sets $K_1$ and $K_2$ is defined as
$d(K_1,K_2)=\min_{x\in K_1, y\in K_2} \abs{x-y}$. Similarly,
$d(x,K)=d(\{x\},K)$. For a set $K$, we use $K^\epsilon$ to denote $\{x\in K\, :\, d(x,\partial K)\ge \epsilon\}$. Finally, for distributions $P$ and $Q$, we use $d_{tv}(P,Q)$ to denote the total variation distance between $P$ and $Q$. 

We will also need some geometric quantities. We will denote lines
(i.e., 1-dimensional affine spaces) by $\ell$. For $x,y\in K$, we
denote their convex hull, that is, the line segment between them,
by $[x,y]$ and $\ell(x,y)$ the line that passes through $x$ and $y$ (which contains $[x,y]$). We also write $[x_1,\ldots, x_k]$ to denote that $x_1,\ldots, x_k$ are collinear.

We also use $\ell_K(x,y)$ to denote the longest connected chord
through $x\in K$ and $y\in K$ contained in $K$ and $|\ell_K(x,y)|$ its
length. We use $a(x,y)$ and $b(x,y)$ to denote the endpoints of
$\ell_K(x,y)$ that are closer to $x$ and $y$, respectively, so that
$\ell_K(x,y) = [a(x,y), b(x,y)] = [b(y,x), a(y,x)]$. The
Euclidean ball of unit radius centered at the origin,
$B_n(0,1)\subset\Real^n$, has volume $\pi_n$. We use $x_{1:m}$ to denote the sequence $x_1,\dots,x_m$. Finally we use $a \wedge
b$ to denote $\min(a,b)$.

\section{Sampling from Non-Convex Spaces}

Most of the known results for the sampling times of the \HNR~exist for
convex sets only. We will think of $\Sigma$ as the image of some
convex set $\Omega$ under a measure preserving, bilipschitz function
$g$. The goal is to understand the relevant geometric quantities of
$\Sigma$ through properties of $g$ and geometric properties of
$\Omega$. We emphasize that the  existence of the map $g$ and its
properties are necessary for the analysis, but the actual algorithm
does not need to know $g$. We formalize this assumption below as well
as describe how we interact with $\Sigma$ and present a few more
technical assumptions required for our analysis. We then present our
main result, and follow that with some conductance results before
moving on to the proof of the theorem in the next section.

\begin{ass}[Oracle Access]
Given a point $u$ and a line $\ell$ that passes through $u$, the oracle returns whether $u\in\Sigma$, and, if so, the largest connected interval in $\ell\cap\Sigma$ containing $u$.
\end{ass}
\begin{ass}[Bilipschitz Measure-Preserving Embeddings]
\label{ass:mapping}
There exist a convex set $\Omega\subset\Real^n$ and a bilipschitz, measure-preserving map $g$ such that  $\Sigma$ is the image of $\Omega$ under $g$. That is, there exists a function $g$ with $\abs{D_g(x)}=1$ (i.e. the Jacobian has unit determinant) with constants $L_\Sigma$ and $L_\Omega$ such that, for any $x,y\in\Omega$, 
\[
\frac{1}{L_\Omega} \abs{x-y} \le \abs{g(x)- g(y)} \le L_\Sigma \abs{x-y}.
\]
In words, $g$ is measure-preserving, $g$ is $L_\Sigma$-Lipschitz, and
$g^{-1}$ is $L_\Omega$-Lipschitz.
\end{ass}
As an example, \cite{Fonseca-Parry-1992} shows that for any star-shaped space, a smooth measure-preserving embedding exists. 
One interesting consequence of Assumption~\ref{ass:mapping} is that
because the mapping is measure-preserving, there must exist a pair
$x,y\in\Omega$ such that $\abs{g(x)- g(y)} \ge \abs{x-y}$. Otherwise,
$\int_\Omega g\le 1$, a contradiction. Similarly, there must exist a pair $u,v\in\Sigma$ such that $\abs{g^{-1}(u)- g^{-1}(v)} \ge \abs{u-v}$. Thus,  
\beq
\label{eq:LOmegaLower}
L_\Omega,L_\Sigma \ge 1 \;. 
\eeq
To simplify the analysis, we will assume that $\Omega$ is a ball with radius $r$. In what follows, we use $x,y,z$ to denote points in $\Omega$, and $u,v,w$ to denote points in $\Sigma$.
We will also assume that $\Sigma$ has no sharp corners and has a smooth boundary:
\begin{ass}[Low Curvature]
\label{ass:curvature}
For any two dimensional plane $\cH\subset \Real^n$, let $\kappa_\cH$ be the curvature of $\partial\Sigma\cap \cH$ and $\cR_\cH$ be the perimeter of $\partial\Sigma\cap \cH$. We assume that $\Sigma$ has low curvature, i.e. that $\kappa=\sup_\cH \kappa_\cH \cR_\cH$ is finite.  
\end{ass}
Assumption~\ref{ass:mapping} does not imply low curvature, as there exist smooth measure-preserving mappings from the unit ball to a cube~\citep{Griepentrog-Hoppner-Kaiser-Rehberg-2008}. 
\begin{ass}
We assume that the volume of $\Sigma$ is equal to one. We also assume that $\Sigma$ contains a Euclidean ball of radius one.  
\end{ass}
Note that the unit ball has volume less than 1 for $n>12$, so for small dimensional problems, we will need to relax this assumption. 

We motivate the forthcoming technical machinery by demonstrating what it can accomplish. The following theorem is the main result of the paper, and the proof makes up most of Section~\ref{sec:analysis}.
\begin{thm}
\label{thm:sampling:main}
Consider the \HNR~algorithm. Let $\sigma_0$ be the distribution of the initial point given to \HNR, $\sigma_t$ be the distribution after $t$ steps of \HNR, and $\sigma$ be the stationary distribution (which is uniform). Let $M=\sup_A \sigma_0(A)/\sigma(A)$. Let $\epsilon$ be a positive scalar. After 
\[
t\ge C' n^6 \log \frac{M}{\epsilon}
\]
steps, we have $d_{tv}(\sigma_t, \sigma) \le \epsilon$. Here $C'$ is a low order polynomial of $L_\Omega,L_\Sigma,\kappa$. 
\end{thm}

%% file: conductance.tex
\subsection{Fast Mixing Markov Chains}

We rely on the notion of conductance \todov{reference?} as our main technical tool. This section recalls the relevant results.

We say that points $u,v\in\Sigma$ see each other if $[u,v]\subseteq\Sigma$. We use $\view(u)$ to denote all points in $\Sigma$ visible from $u$. 
Let $\ell_{\Sigma}(u,v)$ denote the chord through $u$ and $v$ inside $\Sigma$ and $\abs{\ell_{\Sigma}(u,v)}$ its length. Let $P_u(A)$ be the probability of being in set $A\subset \Sigma$ after one step of \HNR~from $u$ and $f_u$ its density function. 
By an argument similar to the argument in Lemma~3 of \cite{Lovasz-1999}, we can show that
\beq
\label{eq:pdf}
f_u(v) = 2\frac{ \one{v\in \view(u)}}{n \pi_n \abs{\ell_{\Sigma}(u,v)}\cdot\abs{u-v}^{n-1}} \;.
\eeq 
The conductance of the Markov process is defined as 
\[
\Phi = \inf_{A \subset \Sigma} \frac{\int_A P_u(\Sigma\setminus A) d u}{\min (\vol(A), \vol(\Sigma\setminus A))} \;.
\]
We begin with a useful conductance result that applies to general Markov Chains.
\begin{lem}[Corollary~1.5 of \citet{Lovasz-Simonovits-1993}]
\label{lem:convergence}
Let $M=\sup_A \sigma_0(A)/\sigma(A)$. Then for every $A\subset \Sigma$, 
\[
\abs{\sigma_t(A)-\sigma(A)} \le \sqrt{M} \left( 1 - \frac{\Phi^2}{2} \right)^t \;.
\] 
\end{lem}
Proving a lower bound on the conductance is therefore a key step in the mixing time analysis. Previous literature has shown such lower bounds for  convex spaces. Our objective in the following is to obtain such bounds for more general non-convex spaces that satisfy bilipschtiz measure-preserving embedding and low curvature assumptions.
 \todoa{Maybe Add a comment about how this part of the paper applies for general Markov chains even though they are studying random walks on a convex body.}

As in previous literature, we shall find that the following \emph{cross-ratio distance} is very useful in deriving an \textit{isoperimetric inequality} and a \textit{total variation inequality}.
\begin{defn}
Let $[a,u,v,b]$ be collinear and inside $\Sigma$, such that $a,b\in \partial \Sigma$. Define 
\[
d_\Sigma(u,v) = \frac{\abs{a-b}\abs{u-v}}{\abs{a-u}\abs{v-b}} \;.
\] 
\end{defn}
It is easy to see that $d_\Sigma(u,v) \ge 4 \abs{u-v}/D_\Sigma$. We define the following distance measure for non-convex spaces.
\begin{defn}
  A set $\Sigma$ will be called $\tau$-best if, for any $u,v\in\Sigma$, there exist points $z_1,\ldots,z_{\tau-1}$ such that $[u,z_1], [z_{\tau-1},v]$, and $[z_i,z_{i+1}]$ for $i=1,\ldots,\tau-2$ are all in $\Sigma$; i.e., any two points in $\Sigma$ can be connected by $\tau$ line segments that are all inside $\Sigma$. We define the distance
\begin{align*}
\widetilde d_\Sigma(u,v) &= \inf_{z_{1:\tau-1}\in \Sigma} \big( d_\Sigma(u,z_1) + d_\Sigma(z_1,z_2) + \dots + d_\Sigma(z_{\tau-1}, v) \big) \;,
\end{align*}
and, by extension, the distance between two subsets
$\Sigma_1,\Sigma_2\subset \Sigma$ as $\widetilde
d_\Sigma(\Sigma_1,\Sigma_2) = \inf_{u\in \Sigma_1, v\in \Sigma_2} \widetilde d_\Sigma(u,v)$. 
\end{defn}
The analysis of the conductance  is often derived via an \textit{isoperimetric inequality}. 
\begin{thm}[Theorem 4.5 of \cite{VEMPALA-2005}]
\label{thm:iso-ineq}
Let $\Omega$ be a convex body in $\Real^n$. Let $h:\Omega\ra\Real^+$ be an arbitrary function. Let $(\Omega_1,\Omega_2,\Omega_3)$ be any partition of $\Omega$ into measurable sets. Suppose that for any pair of points $x\in \Omega_1$ and $y\in \Omega_2$ and any point $z$ on the chord of $\Omega$ through $x$ and $y$, $h(z) \le (1/3) \min(1,d_\Omega(x,y))$. Then 
\[
\vol(\Omega_3) \ge \bE_\Omega(h) \min(\vol(\Omega_1), \vol(\Omega_2)) \,,
\]
where the expectation is defined with respect to the uniform distribution on $\Omega$.
\end{thm}
Given an isoperimetric inequality, a total variation inequality is
typically used in a mixing time analysis to lower bound cross-ratio
distances and then lower bound the
conductance. Our approach is similar. We use the embedding assumption
to derive an isoperimetric inequality in the non-convex space
$\Sigma$. Then we relate cross-ratio distance $d_\Omega$ to distance
$\widetilde d_\Sigma$. This approximation is good when the points are
sufficiently far from the boundary. We incur a small error in the
mixing bound by ignoring points that are too close to the boundary.
Finally we use the curvature condition to derive a total variation
inequality and to lower bound the conductance.

%% file: isoperimetric.tex
\subsection{Cross-Ratio Distances}
\label{sec:isoperimetric}

The first step is to show the relationship between cross-ratio distances in the convex and non-convex spaces. We show that these distances are close as long as points are far from the boundary. These results will be used in the proof of the main theorem in Section~\ref{sec:everything} to obtain an isoperimetric inequality in the non-convex space. 
First we define a useful quantity.
\begin{defn}
Consider a convex set $\Omega$ with some subset $\Omega'$ and collinear points $\{a,x,b\}$ with  $a, b\in\partial \Omega$, $x\in\Omega'$, and $\abs{x-b} \le \abs{x-a}$. Let $c$ be a point on $\partial \Omega$. Let $R(a,x,b,c) =\abs{x-b}/\abs{x-c}$. We use $R(\Omega,\Omega')$ to denote the maximum of $R(a,x,b,c)$ over all such points. We use $R_\epsilon$ to denote $R(\Omega,\Omega^\epsilon)$. 
\end{defn}
The following lemma is the main technical lemma, and we use it to express $\widetilde d_\Sigma$ in terms of $d_\Omega$. 
\begin{lem}
\label{eq:ineq}
Let $\epsilon$ be a positive scalar such that $R_\epsilon(1+8R_\epsilon)\geq 2/3$. Let $\{a,x_1,x_2,b\}$ be collinear such that $a$ and $b$ are on the boundary of $\Omega$, $x_1,x_2\in \Omega^\epsilon$, and $\abs{x_1 -a} < \abs{x_2 - a}$. Let $c$ and $d$ be two points on the boundary of $\Omega$. Then
\[
\frac{\abs{b-a}}{\abs{a-x_1}} \cdot \frac{\abs{x_1 - c}}{\abs{c-d}} \cdot \frac{\abs{d-x_2}}{\abs{x_2-b}} \ge \frac{1}{4 R_\epsilon (1+2 R_\epsilon)} \;.
\]
\end{lem}
\begin{proof}
Let 
\[
A = \frac{\abs{b-a}\abs{x_1-x_2}}{\abs{a-x_1}\abs{x_2-b}} \,, \qquad B = \frac{\abs{c-d} \abs{x_1-x_2}}{\abs{c - x_1} \abs{x_2 - d}} \;.
\]
We prove the claim by proving that $A/B \ge 1/(4 R_\epsilon(1+2 R_\epsilon))$. \\
\textbf{\underline{Case 1, $\abs{x_1 - b} \le \abs{x_1 - a}$}: } In this case, $x_1$ and $x_2$ are both on the line segment $[(a+b)/2, b]$. We consider two cases.\\
\textbf{\underline{\underline{Case 1.1}}, $\abs{x_2 - d} \le \abs{x_2 - b}$: } We have that
\beq
\label{eq:eqqq}
\abs{c-d} = \abs{c-x_1} + \abs{x_1-x_2} + \abs{x_2 - d}\le \abs{c-x_1} + \abs{x_1- b} + \abs{x_2 - b} \;.
\eeq
Because $\abs{x_1 -a} < \abs{x_2 - a}$ by the assumption of the lemma, we have $\abs{x_2 -b} < \abs{x_1 - b}$. Also because $\abs{x_1 - b} \le \abs{x_1 - a}$ in Case 1, we have $ \abs{x_1-b}/\abs{c-x_1} \le R_\epsilon$. Thus
\beq
\label{eq:eqqq2}
\frac{ \abs{x_2-b} }{ \abs{c-x_1} } \leq \frac{ \abs{x_1-b} }{ \abs{c-x_1} } \leq R_\epsilon \;.
\eeq
By \eqref{eq:eqqq} and \eqref{eq:eqqq2}, 
\begin{align*}
\frac{\abs{c-d}}{\abs{c-x_1}} \le 1 + \frac{\abs{x_1-b}}{\abs{c-x_1}} + \frac{\abs{x_2 - b}}{\abs{c-x_1}}\le 1 + R_\epsilon +  \frac{\abs{x_2 - b}}{\abs{c-x_1}}\le 1 + 2 R_\epsilon \;.
\end{align*}
We use also that by definition of $R_\epsilon$  $\abs{x_2 - b} \le R_\epsilon \abs{x_2 - d}$. This and the previous result lets us bound 
\begin{align*}
B &\le (1+2R_\epsilon) \frac{\abs{x_1 - x_2}}{\abs{x_2-d}}\,,\\
A &\ge \frac{\abs{b-a}\abs{x_1-x_2}}{R_\epsilon \abs{a-x_1}\abs{x_2-d}} \ge \frac{\abs{x_1-x_2}}{R_\epsilon \abs{x_2-d}}\,,
\end{align*}
and conclude 
\[
\frac{A}{B} \ge \frac{1}{R_\epsilon(1+2 R_\epsilon)} \ge \frac{1}{4 R_\epsilon (1+2 R_\epsilon)}\;.
\]
\textbf{\underline{\underline{Case 1.2}}, $\abs{x_2 - d} > \abs{x_2 - b}$: } \\
\textbf{Case 1.2.1, $\abs{c-x_1} \le \abs{x_2 - d}$: } We have that $\abs{x_2 - b} < \abs{x_1 - b} \le R_\epsilon \abs{x_1 - c}$. 
Thus, 
\begin{align*}
\frac{A}{B} &\ge \frac{\abs{a-b}\abs{x_1 - x_2}}{B R_\epsilon \abs{a-x_1}\abs{x_1 - c}}\\ 
&= \frac{\abs{a-b}}{R_\epsilon \abs{a-x_1}} \cdot \frac{\abs{d-x_2}}{\abs{c-d}}\ge \frac{\abs{d-x_2}}{R_\epsilon \abs{c-d}}\\
&\ge \frac{\abs{d-x_2}}{R_\epsilon (\abs{d - x_2}+\abs{x_2 - x_1} + \abs{x_1 - c})}\\ 
&\ge \frac{\abs{d-x_2}}{R_\epsilon (\abs{d - x_2}+(1+R_\epsilon)\abs{x_1 - c})}\\
&= \frac{1}{R_\epsilon \left(1+(1+R_\epsilon)\frac{\abs{x_1 - c}}{\abs{d-x_2}}\right)}\ge \frac{1}{R_\epsilon (2+R_\epsilon)}\\ 
&\ge \frac{1}{4 R_\epsilon (1+2 R_\epsilon)}\;.
\end{align*}
\textbf{Case 1.2.2, $\abs{c-x_1} > \abs{x_2 - d}$: } As before, we bound $A$ and $B$ separately:
\begin{align*}
B &\le \frac{\abs{c-d} \abs{x_1-x_2}}{\abs{c - x_1} \abs{x_2 - b}}\\ 
&\le \frac{\abs{x_1-x_2}}{ \abs{x_2 - b}} \cdot \frac{\abs{c-x_1} + R_\epsilon \abs{c-x_1} + \abs{x_2 - d}}{\abs{c-x_1}}\\ 
&\le (2+R_\epsilon) \frac{\abs{x_1-x_2}}{ \abs{x_2 - b}} \,,
\end{align*}
and 
\[
A = \frac{\abs{b-a}\abs{x_1-x_2}}{\abs{a-x_1}\abs{x_2-b}} \ge \frac{\abs{x_1-x_2}}{\abs{x_2-b}}\;.
\] 
Putting these together, 
\[
\frac{A}{B} \ge \frac{1}{2 + R_\epsilon} \ge \frac{1}{4 R_\epsilon (1+2 R_\epsilon)}\,,
\] 
where the second inequality holds because $R_\epsilon(1+8R_\epsilon)\geq 2/3$.\\
\textbf{\underline{Case 2, $\abs{x_1 - b} > \abs{x_1 - a}$ and
$\abs{x_2 - b} < \abs{x_2 - a}$}: } In this case, $x_1$ and $x_2$ are
on opposite sides of the point $(a+b)/2$. Let $M$ be a positive constant. We will choose $M=4$ later.  \\
\textbf{\underline{\underline{Case 2.1}}, $\abs{c-d} \le M \abs{c-x_1}$: } We bound 
\[
B \le \frac{M \abs{x_1 - x_2}}{\abs{x_2 - d}} \le \frac{M R_\epsilon \abs{x_1 - x_2}}{\abs{x_2 - b}}
\]
and conclude
\[
\frac{A}{B} \ge \frac{\abs{a-b}}{M R_\epsilon \abs{a-x_1}} \ge \frac{1}{M R_\epsilon} \ge \frac{1}{4 R_\epsilon (1+2 R_\epsilon)} \;.
\]
\textbf{\underline{\underline{Case 2.2}}, $\abs{c-d} > M \abs{c-x_1}$: }\\
\textbf{Case 2.2.1, $\abs{c-d} \le M \abs{a-b}$: } We have that
\[
\frac{A}{B} \ge \frac{1}{M R_\epsilon^2} \ge \frac{1}{4 R_\epsilon (1+2 R_\epsilon)} \;.
\]
\textbf{Case 2.2.2, $\abs{c-d} > M \abs{a-b}$: } Let $x_0$ be a point on the line segment $[x_1,x_2]$. Let $\beta_1$ be the angle between line segments $[c,x_1]$ and $[x_1,x_0]$. We write
\begin{align*}
\abs{c-x_0}^2 &= \abs{x_1 - x_0}^2 + \abs{x_1 - c}^2\\ 
&\qquad- 2 \abs{x_1 - c}\cdot \abs{x_1 - x_0} \cos \beta_1 \\
&\le \frac{1}{M^2} \abs{c-d}^2 + \frac{1}{M^2} \abs{c-d}^2 + \frac{2}{M^2} \abs{c-d}^2\\ 
&= \frac{4}{M^2} \abs{c-d}^2 \;. 
\end{align*}
By the triangle inequality, 
\[
\abs{d - x_0} \ge \abs{c - d} - \abs{c - x_0}\ge \left( 1 - \frac{2}{M} \right) \abs{c-d} \;.
\] 
Let $\beta_2$ be the angle between line segments $[d,x_2]$ and $[x_2,x_0]$. Let $w=1 - 2/M$. We write
\begin{align*}
w^2 \abs{c-d}^2 &\le \abs{d-x_0}^2 \\
&= \abs{x_2 - x_0}^2 + \abs{x_2 - d}^2- 2 \abs{x_2 - d}\cdot \abs{x_2 - x_0} \cos \beta_2 \\
&\le  \frac{1}{M^2} \abs{c-d}^2+\abs{x_2 - d}^2 + \frac{2}{M} \abs{d - x_2} \cdot \abs{c-d} \;.
\end{align*}
Thus,
\begin{align*}
&\abs{x_2 - d}^2 + \frac{2}{M} \abs{d - x_2} \cdot \abs{c-d}+ \left( \frac{4}{M} - \frac{3}{M^2} -1 \right) \abs{c-d}^2 \ge 0 \,,
\end{align*}
which is a quadratic inequality in $\abs{x_2 - d}$. Thus it holds that 
\[
\abs{x_2 - d} \ge \left( -\frac{1}{M} + \abs{\frac{2}{M}-1} \right) \abs{c-d} \;.
\] 
If we choose $M=4$, then $\abs{x_2 - d} \ge 0.25 \abs{c-d}$ and
\[
B \le \frac{4\abs{x_1 - x_2}}{ \abs{x_1 - c}} \le \frac{4 R_\epsilon \abs{x_1 - x_2}}{\abs{x_1 - a}}\;,
\]
 yielding
\[
\frac{A}{B} \ge \frac{\abs{a-b}}{ 4 R_\epsilon \abs{b-x_2}} \ge \frac{1}{4 R_\epsilon} \ge \frac{1}{4 R_\epsilon (1+2 R_\epsilon)}\;.
\]
Finally, observe that Case~3 follows by symmetry from Case~1.
\end{proof}
The following lemma states that the distance $d_\Omega$ does not increase by adding more steps.
\begin{lem}
\label{lem:triangle-ineq}
Let $a,y_1,y_2,\dots,y_m, b$ be in the convex body $\Omega$ such that  the points $\{a,y_1,y_2,\dots,y_m, b\}$ are collinear. Further assume that $a,b\in\partial \Omega$. We have that
\[
d_\Omega(y_1,y_2) + \dots + d_\Omega(y_{m-1},y_m) \le d_\Omega(y_1,y_m) \;.
\]
\end{lem}
\begin{proof}
We write
\begin{align*}
d_\Omega(y_1,y_m) &= \frac{\abs{a-b}\abs{y_1-y_m}}{\abs{a-y_1}\abs{y_m-b}}\\
&= \frac{\abs{a-b}\abs{y_1-y_2}}{\abs{a-y_1}\abs{y_m-b}} + \frac{\abs{a-b}\abs{y_2-y_3}}{\abs{a-y_1}\abs{y_m-b}}+ \dots + \frac{\abs{a-b}\abs{y_{m-1}-y_m}}{\abs{a-y_1}\abs{y_m-b}}\\
&\ge \frac{\abs{a-b}\abs{y_1-y_2}}{\abs{a-y_1}\abs{y_2-b}} + \frac{\abs{a-b}\abs{y_2-y_3}}{\abs{a-y_2}\abs{y_3-b}}+ \dots + \frac{\abs{a-b}\abs{y_{m-1}-y_m}}{\abs{a-y_{m-1}}\abs{y_m-b}}\\ 
&= d_\Omega(y_1,y_2) + \dots + d_\Omega(y_{m-1},y_m) \;.
\end{align*}
\end{proof}
The next lemma upper bounds $\widetilde d_\Sigma$ in terms of $d_\Omega$. 
\begin{lem}
\label{lem:lipschitz}
Let $x_1,x_2\in \Omega^\epsilon$. We have that
\[
\widetilde d_\Sigma(g(x_1),g(x_2)) \le 4 L_\Sigma^2 L_\Omega^2 R_\epsilon (1+2 R_\epsilon) d_\Omega(x_1,x_2) \;.
\]
\end{lem}
\begin{proof}
First we prove the inequality for the case that $g(x_1)\in \view(g(x_2))$. Let $a,b\in\partial\Omega$ be such that the points $\{a,x_1,x_2,b\}$ are collinear. Let $c,d\in \Omega$ be points such that the points  $\{g(c),g(x_1),g(x_2),g(d)\}$ are collinear and the line connecting $g(c)$ and $g(d)$ is inside $\Sigma$. By the Lipschitzity of $g$ and $g^{-1}$ and Lemma~\ref{eq:ineq},
\begin{align}
\notag
\widetilde d_\Sigma(g(x_1),g(x_2)) &= \frac{\abs{g(c)-g(d)}\abs{g(x_1)-g(x_2)}}{\abs{g(c)-g(x_1)}\abs{g(x_2)-g(d)}}\\
\notag
&\le \frac{L_\Sigma^2 L_\Omega^2 \abs{c-d}\abs{x_1-x_2}}{\abs{c-x_1}\abs{x_2-d}} \\
\notag
&\le L_\Sigma^2 L_\Omega^2 4 R_\epsilon (1+2 R_\epsilon) \frac{\abs{a-b}\abs{x_1-x_2}}{\abs{a-x_1}\abs{x_2-b}}\\
\label{eq:init-ineq}
&= 4 L_\Sigma^2 L_\Omega^2 R_\epsilon (1+2 R_\epsilon) d_\Omega(x_1,x_2) \;.
\end{align}
Now consider the more general case where $g(x_1)\notin \view(g(x_2))$. Find a set of points $y_1,\dots,y_\tau$ such that the line segments $[g(x_1),g(y_1)], [g(y_1),g(y_2)],\dots,[g(y_\tau), g(x_2)]$ are all inside $\Sigma$. By definition of $\widetilde d_\Sigma$, \eqref{eq:init-ineq}, and Lemma~\ref{lem:triangle-ineq}, $\widetilde d_\Sigma(g(x_1),g(x_2))$ can be upper bounded by
\begin{align*}
\inf_{u_{1:\tau-1}\in \Sigma} &d_\Sigma(g(x_1),u_1)+ d_\Sigma(u_1,u_2) + \dots + d_\Sigma(u_{\tau-1}, g(x_2)) \\
&\le d_\Sigma(g(x_1),g(y_1)) + d_\Sigma(g(y_1), g(y_2))+ \dots + d_\Sigma(g(y_{\tau-1}), g(x_2)) \\
&\le 4 L_\Sigma^2 L_\Omega^2 R_\epsilon (1+2 R_\epsilon) (d_\Omega(x_1,y_1)+d_\Omega(y_1,y_2)+\dots+d_\Omega(y_{\tau-1}, x_2)) \\
&\le 4 L_\Sigma^2 L_\Omega^2 R_\epsilon (1+2 R_\epsilon) d_\Omega(x_1,x_2)
\end{align*}
\end{proof}

%% file: total-variation.tex
\subsection{Total Variation Inequality}
\label{sec:total_variation}

In this section, we show that if two points $u,v\in\Sigma$ are close to each other, then $P_u$ and $P_v$ are also close. First we show that if the two points are close to each other, then they have similar views. 
\begin{lem}[Overlapping Views]
\label{lem:view}
Given the curvature $\kappa$ defined in Assumption~\ref{ass:curvature}, for any $u,v\in\Sigma^\epsilon$ such that $\abs{u-v}\le \epsilon' \le \epsilon$,
\[
P_u(\{x : x\notin \view(v)\}) \le \max\left( \frac{4}{\pi},  \frac{\kappa }{\sin(\pi/8)} \right) \frac{\epsilon'}{\epsilon} \;.
\]
\end{lem}
The proof is in Appendix~\ref{app:proofs}. Next we define some notation and show some useful inequalities. For $u
\in \Sigma$, let $w$ be a random point obtained by making one step of
\HNR~from $u$. Define $F(u)$ by $\Prob{\abs{w-u} \le F(u)}=1/8$. If
$d(u, \partial\Sigma)\ge h$, less than $1/8$ of any chord passing
through $u$ is inside $B(u,h/16)$. Thus $\Prob{\abs{u-w}\le h/16} \le
1/8$, which implies 
\beq
\label{eq:Fu}
F(u) \ge \frac{h}{16} \;.
\eeq
Intuitively, the total variation inequality implies that if $u$ and $v$ are close geometrically, then their proposal distributions must be close as well. 
\begin{lem}
\label{lem:total-var}
Let $u,v\in \Sigma^\epsilon$ be two points that see each other. Let $\epsilon' = \frac{\epsilon}{6} \min\left( \frac{\pi}{4},
\frac{\sin(\pi/8)}{\kappa} \right)$. Suppose that 
\[
d_\Sigma(u,v) < \frac{\epsilon}{24 D_\Sigma} \qquad \text{ and }\qquad \abs{u-v} < \min\left(\frac{2 F(u)}{\sqrt{n}}, \epsilon' \right) \;. 
\]
Then, 
\[
\abs{P_u - P_v} < 1- \frac{\epsilon}{8 e^4 D_\Sigma} \;.
\] 
\end{lem}
The proof is in Appendix~\ref{app:proofs}. The proof uses ideas from proof of Lemma~9 of \cite{Lovasz-1999}. The proof of \cite{Lovasz-1999} heavily relies on the convexity of the space, which does not hold in our case. We overcome the difficulties using the low curvature assumption and the fact that $u$ and $v$ are sufficiently far from the boundary.

%% file: final.tex
\subsection{Putting Everything Together}
\label{sec:everything}

Next we bound the conductance of \HNR. 
\begin{lem}
\label{lem:s-cond}
Let 
\[
\delta=\frac{9r}{320 e^4 n L_\Omega D_\Sigma}\,,\qquad G = \frac{1}{6}\min\left( \frac{\pi}{4}, \frac{\sin(\pi/8)}{\kappa} \right)\,,\qquad \epsilon' = \frac{9r}{20n}\,,\qquad N=\frac{9r}{80n L_\Sigma^2 L_\Omega^3 R_{\epsilon'} (1+2 R_{\epsilon'})}\,,
\]
where $r$ is the radius of ball $\Omega$
(so $r^n \pi_n=1$). The conductance $\Phi$ of \HNR~is at least 
\[
\frac{\delta}{4}\left( \frac{2}{5 n D_\Omega} \wedge N \left( \frac{1}{24 D_\Sigma} \wedge \frac{2}{\sqrt{n}} \left( \frac{1}{8\sqrt{n}} \wedge G \right) \right) \right) \;.
\]
\end{lem}
The proof is in Appendix~\ref{app:proofs}. In proving this lemma, the non-convexity of $\Sigma$ is specially troubling when points are close to the boundary. We overcome this difficulty by using the isoperimetric inequality shown in Theorem~\ref{thm:iso-ineq}, which is in terms of average distances instead of minimum distances. This enables us to ignore points that are very close to the boundary. 


If we treat $L_\Sigma, L_\Omega, \kappa$ as constants and collect all constants in $C$, we have a $\Phi\ge C/n^3$ lower bound for the conductance. Now we are ready to prove the main theorem. 
\begin{proof}[Proof of Theorem~\ref{thm:sampling:main}]
Using Lemma~\ref{lem:convergence} and Lemma~\ref{lem:s-cond}, $d_{tv}(\sigma_t, \sigma) \le \sqrt{M} ( 1 - C^2/(2 n^6) )^t$, 
which gives the final bound after rearrangement.
\end{proof}

%% file: experiments.tex
\section{Planning}
\label{sec:experiments}

This section makes an empirical argument for use of the \HNR~in trajectory planning. In the first of two experiments, the state space is a position vector constrained to some map illustrated by the bottom plots of Figure~\ref{fig:position.planning}. The second experiment also includes two dimensions of velocity in the state and limits state transitions to those that respect the map as well as kinematics and requires the planning to control the system explicitly (by specifying an acceleration vector for every time step). We will show that \HNR~outperforms RRT in both cases by requiring fewer transitions to reach the goal state across a wide variety of map difficulties.

\subsection{Position only}
The state starts at the bottom left of the spiral and the goal is the top right. Both algorithms are implemented as described in the introduction. The number of tranitions needed to reach the goal of both algorithms is plotted as a function of the width of the spiral arms; the larger the width, the easier the problem. 

The results are presented in Figure~\ref{fig:position.planning}. The top plot show the number of transitions needed by both algorithms as the width of the arms changes, averaged over 500 independent runs. We see that the \HNR~outperforms RRT for all but the hardest problems, usually by a large margin. The two lower plots show the sample points produced from one run with width equal to 1.2; we see that RRT has more uniform coverage, but that \HNR~has a large speedup over linear sections, therefore justifying its faster exploration.
\begin{figure}[t]
  \centering
  \includegraphics[width=.47\textwidth,natwidth=610,natheight=642]{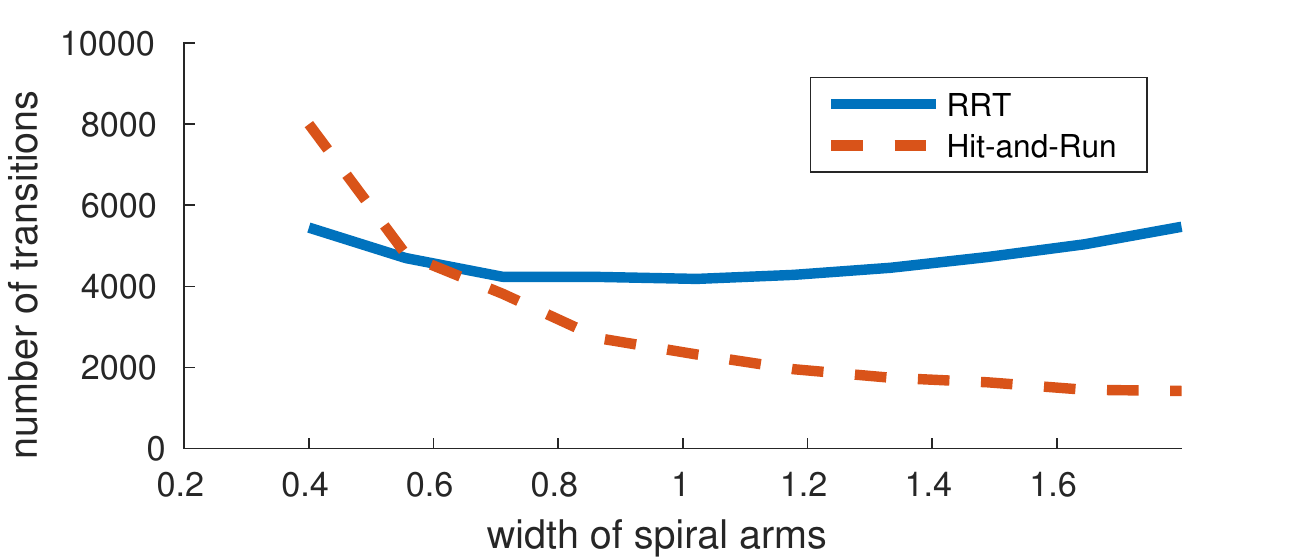}
  \includegraphics[clip=true, trim=1cm 0cm 0cm 0cm, width=.5\textwidth,natwidth=610,natheight=642]{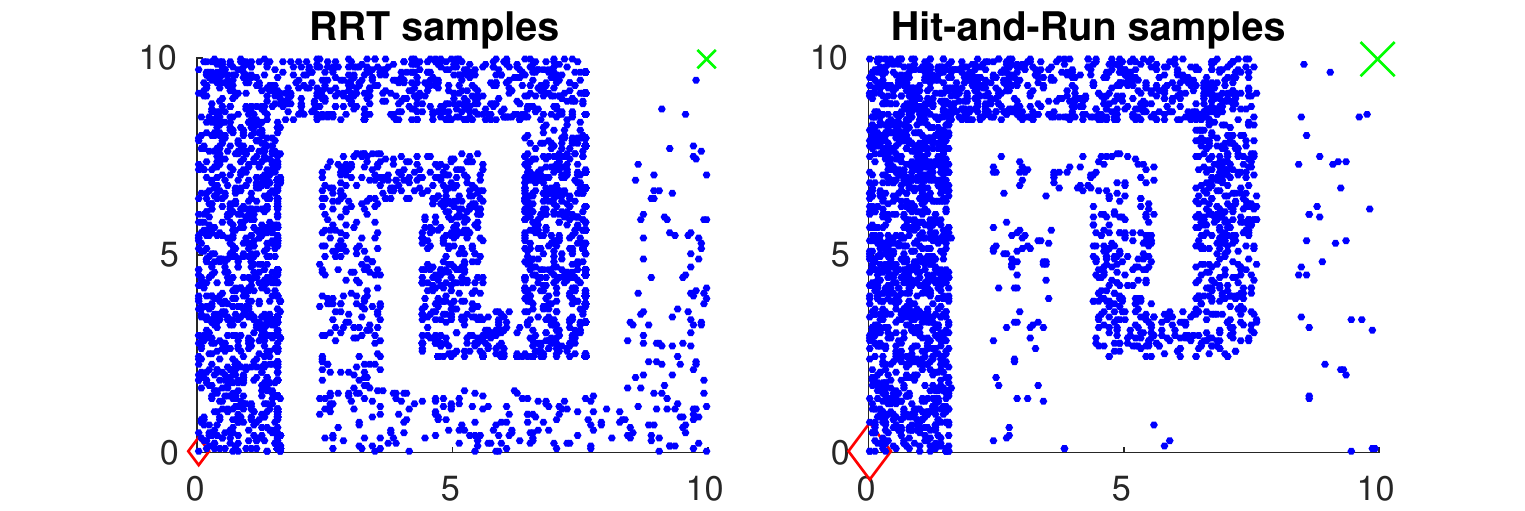}
  \caption{Position only planning example}
  \label{fig:position.planning}
\end{figure}

RRT is slow in this problem because in many rounds the tree does not grow in the
\textit{right} direction. For example at the beginning the tree needs
to grow upwards, but most random samples will bias the growth to
right. As \HNR~only considers the space that is visible to the current
point, it is less sensitive to the geometry of the free space. We can
make this problem arbitrarily hard for RRT by making the middle part
of the spiral fatter. \HNR, on the other hand, is insensitive to such
changes. Additionally, the growth of the RRT tree can become very slow towards the end. This is because the rest of the tree absorbs most samples, and the tree grows only if the random point falls in the vicinity of the goal.

\subsection{Kinematic Planning}
In this set of simulations, we constrain the state transitions to adhere to the laws of physics: the state propagates forward under kinematics until it exits the permissible map, in which case it stops inelastically at the boundary. The position map is the two-turn corridor, illustrated in the bottom plots of Figure~\ref{fig:kinematic.planning}. Both algorithms propose points to in the analogous manner to the previous section (where a desired speed is sampled in addition to a desired position); then, the best acceleration vector in the unit ball is calculated and the sample is propagated forward by the kinematics. If the sample point encounters the boundary, the velocity is zeroed. Both RRT and \HNR~are constrained to use the same controller and the only difference is what points are proposed. 
\begin{figure}
  \centering
  \includegraphics[clip=true, trim=0cm 0cm .5cm 0cm, width=.47\textwidth,natwidth=610,natheight=642]{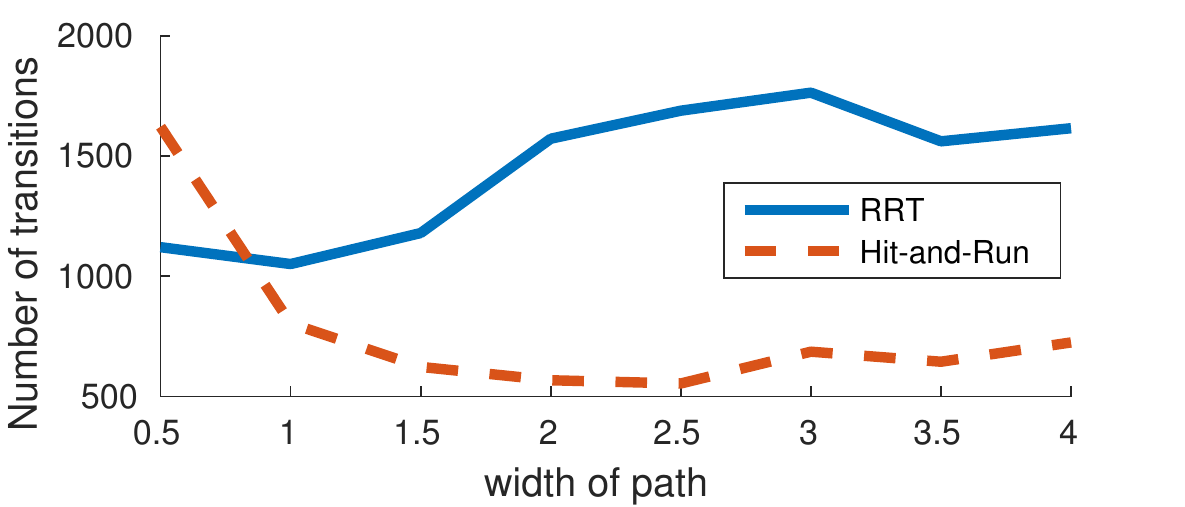}
  \includegraphics[clip=true, trim=1cm 0cm 0cm 0cm, width=.5\textwidth,natwidth=610,natheight=642]{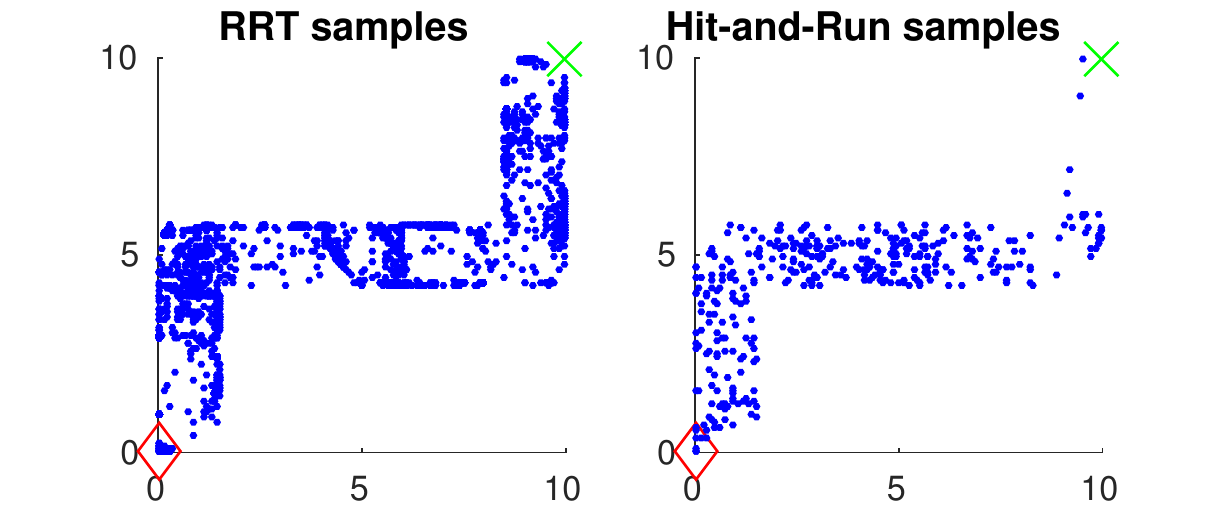}
  \vspace{-.5cm}
  \caption{Performance under kinematic constraints}
  \label{fig:kinematic.planning}
\end{figure}

We see that \HNR~again outperforms RRT across a large gamut of path widths by as much as a factor of three. The bottom two plots are of a typical sample path, and we see that \HNR~has two advantages: it accelerates down straight hallways, and it samples more uniformly from the state space. In contrast, RRT wastes many more samples along the boundaries.

%% file: proofs.tex
\section{Proofs}
\label{app:proofs}

\begin{figure*}
\centering

\begin{subfigure}[t]{0.4\textwidth}
  \centering
  \includegraphics[width=\textwidth,natwidth=610,natheight=642]{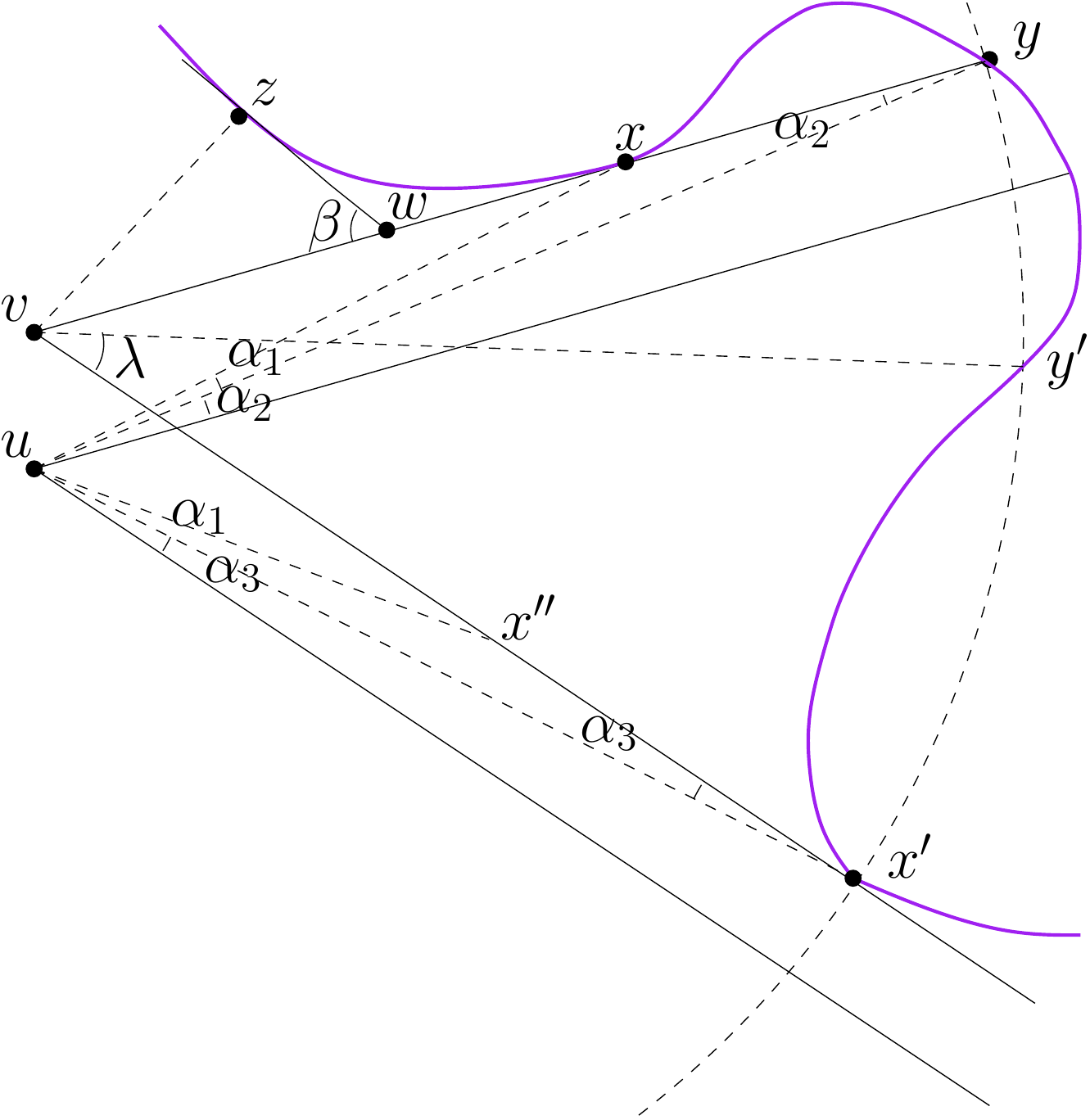}
    \caption{The geometry used in Lemma~\ref{lem:view}.}
  \end{subfigure}
  ~
  \begin{subfigure}[t]{0.4\textwidth}
  \centering
  \includegraphics[width=\textwidth,natwidth=610,natheight=642]{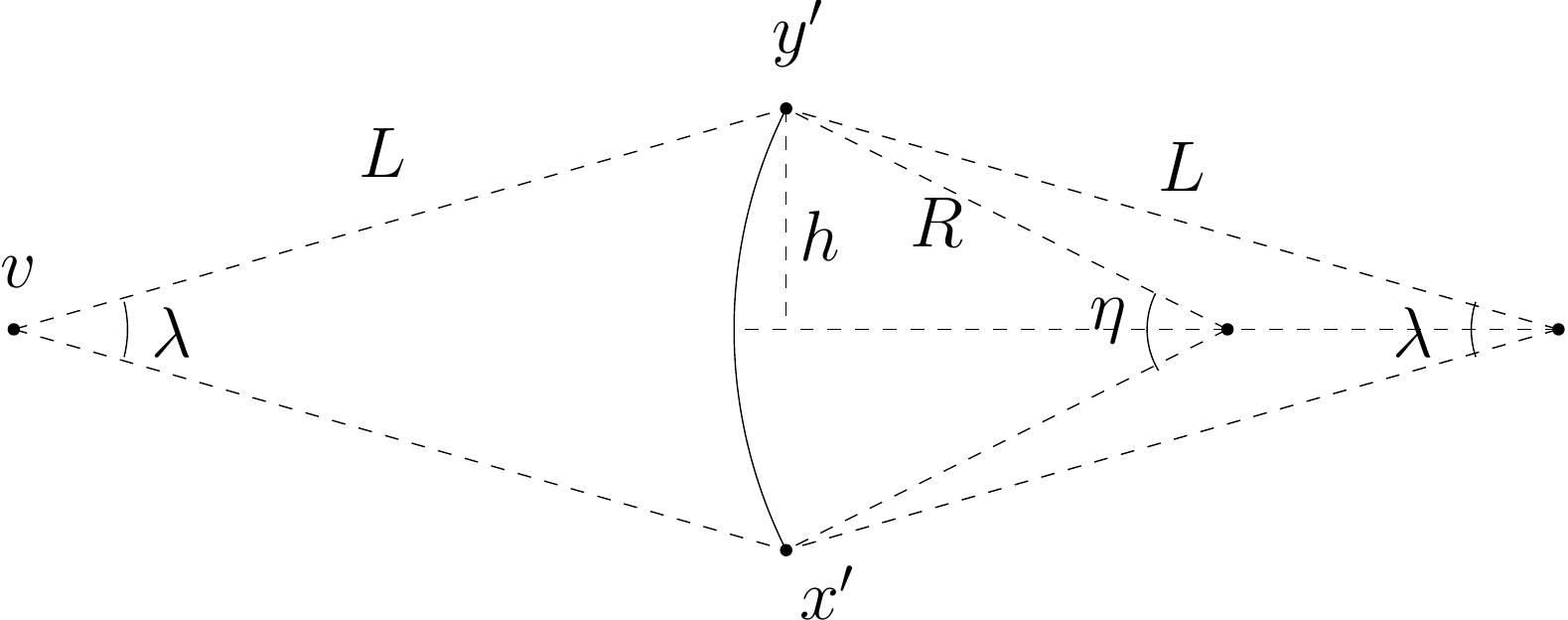} 
    \caption{The geometry used in Lemma~\ref{lem:view}.}
  \end{subfigure}
  
\caption{If two points are close to each other, they have similar views.}  
\label{fig:visible}
\end{figure*}

\begin{proof}[Proof of Lemma~\ref{lem:view}]
We say a line segment $L$ is not fully visible from a point $x$ if
there exists a point on the line segment that is not visible from $x$.
We denote this event by $L\notin \view(x)$. Let $L$ be a line segment
chosen by \HNR~from $u$. So, as the next point in the Markov chain,
\HNR~chooses a point uniformly at random from $L$. We know that 
\[
P_u(\{x : x\notin \view(v)\}) \le P_u(\{L : L\notin \view(v)\})\,,
\] 
So it suffices to show
\beq
\label{eq:line-visible}
P_u(\{L : L\notin \view(v)\}) \le  \max\left( \frac{4}{\pi},  \frac{\kappa }{\sin(\pi/8)} \right) \frac{\epsilon'}{\epsilon} \;.
\eeq
To sample the line segment $L$, first we sample a random two dimensional plane containing $u$ and $v$, and then sample the line segment inside this plane. To prove \eqref{eq:line-visible}, we show that in any two dimensional plane containing $u$ and $v$, the ratio of invisible to visible region is bounded by $\max\left( \frac{4}{\pi},  \frac{\kappa }{\sin(\pi/8)} \right) \frac{\epsilon'}{\epsilon}$.\todov{be sure the plane argument is not fishy}

Consider the geometry shown in Figure~\ref{fig:visible}(a). Let $\cH$ be the intersection of $\partial\Sigma$ and a two dimensional plane containing $u$ and $v$. For a line $\ell$ and points $q$ and $u$, we write $[q,\ell,u]$ to denote that $u$ and a small neighborhood of $q$ on $\cH$ are on the opposite sides of $\ell$. For example, in Figure~\ref{fig:visible}(a), we have that $[x,\ell(v,x),u]$. Define a subset 
\[
Q = \{q\in\cH\,:\, \ell(v,q) \mbox{ is tangent to } \cH \mbox{ at } q \mbox{ and } [q,\ell(v,q),u]\} \;.
\]
Any line $\ell(v,q)$ such that $[q,\ell(v,q),u]$ creates some space that is visible to $u$ and invisible to $v$. If $Q$ is empty, then the entire $\cH$ is in the view of $v$ and $P_u(\{x : x\notin \view(v)\}) = 0$.\todov{Can you make this claim formal?} Otherwise, let $x$ be a member of $Q$. Let $y\in\cH$ be the closest point to $x$ such that $[v,y]$ is tangent to $\cH$ at $x$. Let $\alpha_1$  be the angle between $[x,u]$ and $[u,y]$, and let $\alpha_2$ be the angle between $[y,v]$ and $[y,u]$. Because $\abs{u-v}\le \abs{v-z}\le \abs{v-x}$, $\alpha_1+\alpha_2\le \pi/2$. Further, if the lengths of $\abs{u-v}$ and $\abs{v-x}$ are fixed, $\alpha_1+\alpha_2$ is maximized when $[v,u]$ is orthogonal to $[u,x]$. If $x$ is the only member of $Q$, then maximum invisible angle is $\alpha_1$, which can be bounded as follows: 
\[
\sin\alpha_1 \le \sin(\alpha_1 + \alpha_2) \le \frac{\abs{u-v}}{\abs{v-x}} \le \frac{\abs{u-v}}{\epsilon}\;.
\]
Otherwise, assume $Q$ has more members. The same upper bound holds for members that are also on the line $\ell(v,x)$. So next we consider members of $Q$ that are not on the line $\ell(v,x)$. Assume $Q$ has only one such member and let $x'$ be that tangent point (see Figure~\ref{fig:visible}(a). The same argument can be repeated if $Q$ has more such members). We consider two cases. \textbf{Case 1:} $\abs{v-x'} \ge \abs{v-y}$. Let $\alpha_3$ be the angle between $[v,x']$ and $[u,x']$. If $\alpha_3 \le \alpha_2$, then
\[
\sin (\alpha_1+\alpha_3) \le \sin (\alpha_1+\alpha_2) \le  \frac{\abs{u-v}}{\epsilon} \;. 
\] 
Otherwise, $\alpha_3 > \alpha_2$. Consider point $x''$ such that the angle between $[u,x']$ and $[u,x'']$ is $\alpha_1$. We show that $\abs{v-x''} \ge \abs{v-x}$ by contradiction. Assume $\abs{v-x} > \abs{v-x''}$. Thus, $\abs{x'-x''} > \abs{y-x}$ and $\abs{u-x} > \abs{u-x''}$. By law of sines, $\abs{x-y}/\sin \alpha_1 = \abs{u-x}/\sin \alpha_2$ and $\abs{x''-x'}/\sin \alpha_1 = \abs{u-x''}/\sin \alpha_3$. Because $\abs{u-x} > \abs{u-x''}$ and $\alpha_3 > \alpha_2$, we have that $\abs{u-x}/\sin \alpha_2 > \abs{u-x''}/\sin \alpha_3$, and thus $\abs{x''-x'}/\sin \alpha_1 < \abs{x-y}/\sin \alpha_1$. This implies $\abs{x''-x'} < \abs{x-y}$, a contradiction. \todov{add argument to have u on the right side of v}\todoy{the new definition of $Q$ ensures that $u$ is in the right side.}
Thus, 
\[
\sin (\alpha_1+\alpha_3) \le \frac{\abs{u-v}}{\abs{v-x''}} \le \frac{\abs{u-v}}{\abs{v-x}} \le \frac{\abs{u-v}}{\epsilon}\;. 
\]  

Next we consider the second case. \textbf{Case 2:} $\abs{v-x'} < \abs{v-y}$. Consider the arc on $\cH$ from $y$ to $x'$. Let $y'$ be the last point on this arc such that $\abs{v-y'} = \abs{v-y}$. Let $\eta$ be the change of angle between the tangent of $\cH$ at $y'$ and the tangent of $\cH$ at $x'$ (tangents are defined in clockwise direction), and let $\lambda$ be the angle between $[v,y']$ and $[v,x']$. Angle $\eta$ is minimized
\todoa{I don't see why this is true}\todoy{Notice that $y'$ is the last point on this arc such that $\abs{v-y'} = \abs{v-y}$. So the angle between $[y',v]$ and the tangent at $y'$ is less than $\pi/2$.}\todov{define angles between lines properly as two opposite choices can be made } when the tangent at $y'$ is orthogonal to $[v,y']$. Thus $\eta \ge \pi/2 - \lambda$. If $\lambda < \pi/4$, then $\eta \ge \pi/4$. Angle $\lambda$ is smallest when the arc from $y'$ to $x'$ changes with maximum curvature $\kappa/\cR_\cH$, i.e. it is a segment of a circle with radius $\cR_\cH/\kappa$. Figure~\ref{fig:visible}(b) shows this case, where $R=\cR_\cH/\kappa$ and $L = \abs{v-y'}$. We have that
\[
\frac{\sin(\lambda/2)}{\sin(\eta/2)} \ge \frac{h/L}{h/R} = \frac{R}{L} = \frac{\cR_\cH}{\kappa \abs{v-y}}
\]
Thus, 
\[
\frac{\lambda}{2} \ge \sin (\lambda/2) \ge \frac{\cR_\cH}{\kappa \abs{v-y}} \sin(\eta/2) = \frac{\sin(\pi/8) \cR_\cH}{\kappa \abs{v-y}} \ge \frac{\sin(\pi/8)}{\kappa}  \,,
\]
where the last step follows by $\abs{v-y} \le \cR_\cH$. Thus 
\[
\lambda \ge \lambda_0 \eqdef \min\left(\frac{\pi}{4}, \frac{\sin(\pi/8)}{\kappa} \right) \;.
\]
So for every $\abs{u-v}/\epsilon$ invisible region, we have at least $\lambda_0$ visible region. Thus,
\[
P_u(\{x : x\notin \view(v)\}) \le \frac{\abs{u-v}}{\lambda_0 \epsilon} = \max\left( \frac{4}{\pi},  \frac{\kappa }{\sin(\pi/8) } \right) \frac{\abs{u-v}}{\epsilon} \;.
\]
\end{proof}

Before proving Lemma~\ref{lem:total-var}, we show a useful inequality. Consider points $u,v,w\in\Sigma$ that see each other. Let $\cC$ be the convex hull of 
\[
\{a(u,v),b(u,v),a(u,w),b(u,w),a(v,w),b(v,w)\} \;. 
\]
Let $i$ and $j$ be distinct members of $\{u,v,w\}$. We use $a'(i,j)$ and $b'(i,j)$ to denote the endpoints of $\ell_{\cC}(i,j)$ that are closer to $i$ and $j$, respectively. Because $\abs{a'(i,j)-b'(i,j)}$ is convex combination of two line segments that are inside $\Sigma$, 
\beq
\label{eq:chull-bounded}
\abs{a'(i,j)-b'(i,j)} \le D_\Sigma \;. 
\eeq
Also $[a(i,j), b(i,j)]\subset [a'(i,j), b'(i,j)]$, and thus $\abs{a(i,j)-i} \le \abs{a'(i,j)-i}$ and $\abs{b(i,j)-j} \le \abs{b'(i,j)-j}$. We can write
\begin{align}
\notag
\ell_\cC(i,j) &= \frac{\abs{i-j}\cdot\abs{a'(i,j)-b'(i,j)}}{\abs{a'(i,j)-i}\cdot\abs{j-b'(i,j)}} \\
\notag
&\le \frac{\abs{i-j}\cdot\abs{a(i,j)-b(i,j)}}{\abs{a(i,j)-i}\cdot\abs{j-b(i,j)}} \cdot \frac{\abs{a'(i,j)-b'(i,j)}}{\abs{a(i,j)-b(i,j)}} \\
\label{eq:convex-hull-sigma}
&\le D_\Sigma \frac{\ell_\Sigma(i,j)}{d(i,\partial\Sigma)} \,,
\end{align}
where the last inequality holds because $\abs{a(i,j)-b(i,j)} \ge d(i,\partial\Sigma)$.
\begin{proof}[Proof of Lemma~\ref{lem:total-var}]

Let $A\subset \Sigma$ be a measurable subset of $\Sigma$. We prove that
\[
P_u(A) - P_v(A) \le 1- \frac{\epsilon}{8 e^4 D_\Sigma} \;.
\]
\todoa{needs $|\cdot|$?}
\todoy{no!}
We partition $A$ into five subsets, and estimate the probability of each of them separately:
\begin{align*}
A_1 &= \left\{ x\in A\, :\, \abs{x-u} < F(u) \right\} \,,\\
A_2 &= \left\{ x\in A\, :\, \abs{(x-u)^\top (u-v)} > \frac{1}{\sqrt{n}} \abs{x-u}\cdot \abs{u-v} \right\}\,,\\
A_3 &= \Bigg\{ x\in A\, :\, \abs{x-u} <  \frac{1}{6} \abs{u-a(u,x)}\,,\\ 
&\qquad\qquad\qquad\mbox{ or } \abs{x-u} <  \frac{1}{6} \abs{u-a(x,u)} \Bigg\} \,,\\
A_4 &= \left\{ x\in A\, :\, x\in \view(u),\, x\notin \view(v) \right\}\,,\\
S&= A\setminus A_1\setminus A_2\setminus A_3\setminus A_4\;.
\end{align*}

The definition of $F(u)$ immediately yields $P_u(A_1)\le 1/8$. Now consider $A_2$ and let $C$ be the cap of the unit sphere centered at $u$ in the direction of $v$, defined by
$C= \{x:(u-v)^\top x\geq\frac{1}{\sqrt{n}}|u-v|\}$. If $x\sim P_u$, then $P(x\in A_2)$ is bounded above by the probability that a uniform random line through $u$ intersects $C$, which has probability equal to the ratio between the surface of $C$ and the surface of the half-sphere. A standard computation to show that this ratio is less than $1/6$, and hence $P_u(A_2)\leq 1/6$. The probability that $x\in A_3$ is at most $1/6$, since $x$ is chosen from a segment of a chord of length at most $|\ell(u,x)|/6$. Finally, to bound $P(A_4)$, we apply Lemma~\ref{lem:view}:
\begin{equation*}
  P_u\left( x\in A\, :\, x\in \view(u),\, x\notin \view(v) \right) \le \max\left( \frac{4}{\pi},  \frac{\kappa }{\sin(\pi/8)} \right)
  \frac{\epsilon'}{\epsilon}
  \le \frac{1}{6} \;.
\end{equation*}
The combined probability of $A_1,A_2,A_3$, and $A_4$ is at most $1/8 + 1/6 + 1/6 + 1/6 < 3/4$.

We now turn to bounding $P_u(S)$ and show that $P_u(S) \le 2 e^4 (D_\Sigma/\epsilon) P_v(S)$. Because points in $S$ are visible from both $u$ and $v$, by \eqref{eq:pdf} 
\[
P_v(S) = \frac{2}{n \pi_n} \int_S \frac{1}{\ell_\Sigma(v,x) \abs{x-v}^{n-1}} \;.
\] 
Now, any $x\in S$ must respect the following
\begin{align}
\label{eq:eq11}
\abs{x-u} &\ge F(u) \ge \frac{\sqrt{n}}{2} \abs{u-v}\,,\\
\label{eq:eq22}
\abs{(x-u)^\top (u-v)} &\le \frac{1}{\sqrt{n}} \abs{x-u}\cdot \abs{u-v}\,,\\
\label{eq:eq33}
\abs{x-u} &\geq  \frac{1}{6} \abs{u-a(u,x)},\text{ and } \\
\label{eq:eq44}
\abs{x-u} &\geq \frac{1}{6} \abs{u-a(x,u)} \;.
\end{align}
\begin{figure}
\begin{center}	
\includegraphics[scale=.8,natwidth=610,natheight=642]{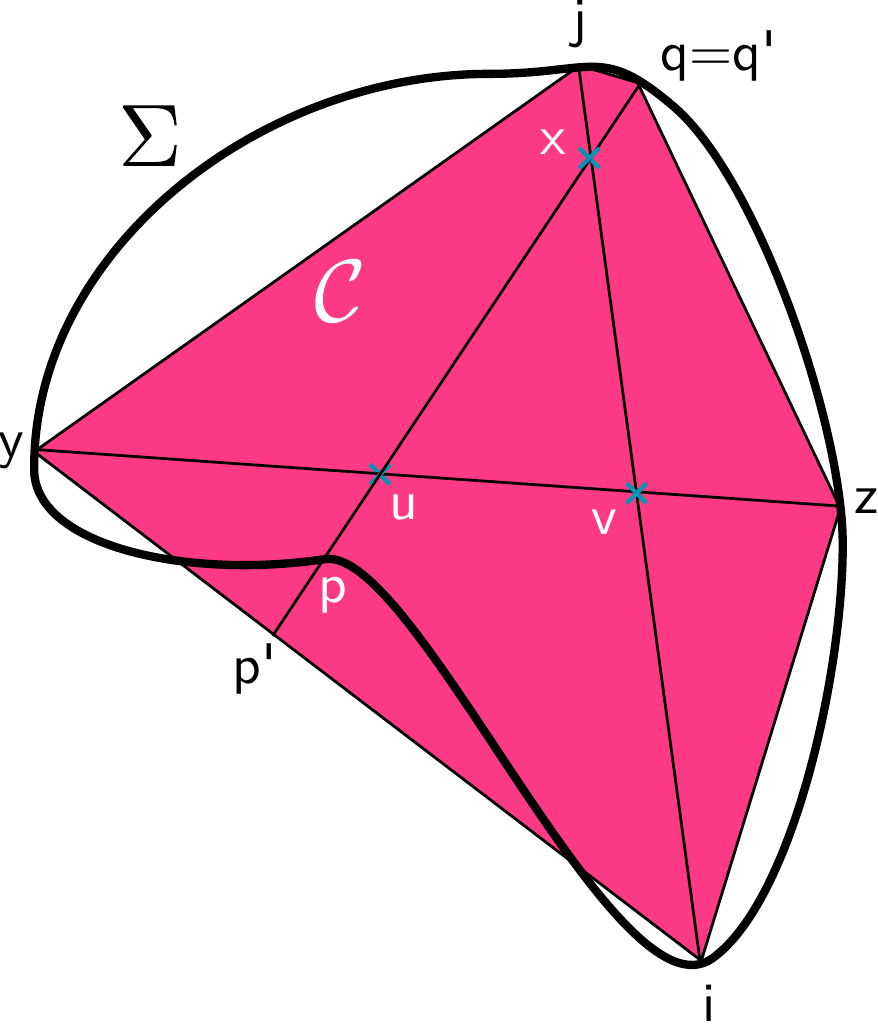}
\end{center}
\caption{Illustration for Lemma 12 proof}
\label{fig:lemma12}
\end{figure}
As illustrated in Figure~\ref{fig:lemma12}, we define the points $y = a(u, v)$, $z = a(v, u)$, $p = a(u, x)$, $q = a(x, u)$, $i = a(v, x)$ and $j = a(x, v)$ with convex hull $\cC$. Also let $p'$ and $q'$ be the endpoints of $\ell_\cC(u,x)$. If $p'=p$ and $q'=q$, we proceed with the argument in the  proof of Lemma~9 of \cite{Lovasz-1999} to get the desired result. Otherwise, assume $q'=q$ and $p'$ is the intersection of the lines $\ell(u,p)$ and $\ell(y,i)$. (See Figure~\ref{fig:lemma12}. A similar argument holds when $q\neq q'$.) From \eqref{eq:eq33} and \eqref{eq:eq44}, we get that 
\[
2 \abs{x-u} > \frac{1}{6}  \abs{p-q} \;.
\]
We have that $\abs{p-q} \ge \epsilon$, and by \eqref{eq:chull-bounded}, $\abs{p' - q'}\le D_\Sigma$. Thus $\abs{p' - q'} \le (D_\Sigma/\epsilon) \abs{p - q}$. Thus,
\beq
\label{eq:xupq}
\frac{1}{6} \abs{p' - q'} \le \frac{2D_\Sigma}{\epsilon} \abs{x-u} \;.
\eeq
To relate $P_v(S)$ to $P_u(S)$, we need to bound $\abs{x-v}$ and $\ell(x,v)$ in terms of $\abs{x-u}$ and $\ell(x,u)$:
\begin{align*}
\abs{x-v}^2 &= \abs{x-u}^2 + \abs{u-v}^2 + 2 (x-u)^\top (u-v) \\
&\le \abs{x-u}^2 + \abs{u-v}^2 + \frac{2}{\sqrt{n}} \abs{x-u}\cdot \abs{u-v} & \text{\dots By \eqref{eq:eq22}}\\
&\le \abs{x-u}^2 + \frac{4}{n} \abs{x-u}^2 + \frac{4}{n} \abs{x-u}^2 & \text{\dots By \eqref{eq:eq11}} \\
&= \left( 1 + \frac{8}{n} \right) \abs{x-u}^2 \;.
\end{align*}
Thus,
\beq
\label{eq:xvxu}
\abs{x-v} \le \left( 1 + \frac{4}{n} \right) \abs{x-u} \;.
\eeq
First we use convexity of $\cC$ to bound $\ell_\cC(x,v)$ in terms of $\ell_\cC(x,u)$, and then we use \eqref{eq:convex-hull-sigma} to bound $\ell_\Sigma(x,v)$ and $\ell_\Sigma(x,u)$ in terms of $\ell_\cC(x,v)$ and $\ell_\cC(x,u)$. 
By Menelaus' Theorem (wrt triangle $uvx$ and transversal line $[y,i]$),
\[
\frac{\abs{x-i}}{\abs{v-i}} = \frac{\abs{u-y}}{\abs{v-y}} \cdot \frac{\abs{x-p'}}{\abs{u-p'}} \;.
\]
We have that
\[
\frac{\abs{u-y}}{\abs{v-y}} = 1 -  \frac{\abs{v-u}}{\abs{v-y}} > 1 - d_\cC(u,v)\,,
\]
and thus
\begin{align*}
 \frac{\abs{x-v}}{\abs{v-i}} &=  \frac{\abs{x-i}}{\abs{v-i}} - 1\\
 &\ge (1-d_\cC(u,v))  \frac{\abs{x-p'}}{\abs{u-p'}} - 1\\
 &=  \frac{\abs{x-u}}{\abs{u-p'}} \left( 1 - d_\cC(u,v)  \frac{\abs{x-p'}}{\abs{x-u}} \right)\\
 &>  \frac{\abs{x-u}}{\abs{u-p'}} \left( 1 - d_\cC(u,v)  \frac{\abs{p'-q'}}{\abs{x-u}} \right)\\
 &>  \frac{\abs{x-u}}{\abs{u-p'}} \left( 1 - \frac{12 D_\Sigma \epsilon}{24 D_\Sigma \epsilon}  \right)\\
 &> \frac{1}{2}  \frac{\abs{x-u}}{\abs{u-p'}}\,,
 \end{align*}
where we have used \eqref{eq:xupq} and $d_\cC(u,v) = d_\Sigma(u,v) < \epsilon/(24 D_\Sigma)$ (the condition in the statement of the lemma); we conclude that 
\begin{equation}
  \label{eq:v.i.bound}
\abs{v-i} < 2 \frac{\abs{x-v}}{\abs{x-u}} \abs{u-p'}  \;.
\end{equation}

Next we prove a similar inequality for $\abs{v-j}$. It is easy to check that
\[
\frac{\abs{z-v}}{\abs{u-z}} = 1 -  \frac{\abs{u-v}}{\abs{u-z}} > 1 - d_\cC(u,v)\,,
\]
and combining with Menelaus' Theorem 
\[
\frac{\abs{v-j}}{\abs{x-j}} = \frac{\abs{q'-u}}{\abs{x-q'}} \cdot \frac{\abs{z-v}}{\abs{u-z}}
\]
we can show 
\begin{align*}
 \frac{\abs{x-v}}{\abs{x-j}} &=  \frac{\abs{v-j}}{\abs{x-j}} - 1\\
 &\ge (1-d_\cC(u,v))  \frac{\abs{q'-u}}{\abs{x-q'}} - 1\\
 &=  \frac{\abs{x-u}}{\abs{x-q'}} \left( 1 - d_\cC(u,v)  \frac{\abs{q'-u}}{\abs{x-u}} \right)\\
 &>  \frac{\abs{x-u}}{\abs{x-q'}} \left( 1 - d_\cC(u,v)  \frac{\abs{p'-q'}}{\abs{x-u}} \right)\\
 &>  \frac{\abs{x-u}}{\abs{x-q'}} \left( 1 - \frac{12 D_\Sigma \epsilon}{24 D_\Sigma \epsilon} \right)\\
 &> \frac{1}{2}  \frac{\abs{x-u}}{\abs{x-q'}}\,,
 \end{align*}
 where we have used \eqref{eq:xupq} and $d_\cC(u,v) = d_\Sigma(u,v) < \epsilon/(24 D_\Sigma)$. Thus,
 \[
 \abs{x-j} < 2 \frac{\abs{x-v}}{\abs{x-u}} \abs{x-q'}  \;,
\]
and combining this with the trivial observation that 
$
\abs{x-v} \le 2 \frac{\abs{x-v}}{\abs{x-u}} \abs{x-u}\,,
$
and Equation~\ref{eq:v.i.bound} yields
\[
  \ell_\cC(x,v) =
  \abs{v-i}+\abs{v-x}+\abs{x-j}
  \le
  2 \frac{\abs{x-v}}{\abs{x-u}} \ell_\cC(x,u)\;.
\]
Thus,
\begin{align}
\notag
\ell_\Sigma(x,v) &= \ell_\cC(x,v)\\ 
\notag
&\le 2 \frac{\abs{x-v}}{\abs{x-u}} \ell_\cC(x,u)\\ 
\label{eq:lxvlxu}
&\le \frac{2 D_\Sigma}{\epsilon} \frac{\abs{x-v}}{\abs{x-u}}\ell_\Sigma(x,u)\;.
\end{align}
Where the last step holds by \eqref{eq:convex-hull-sigma}. Now we are ready to lower bound $P_v(S)$ in terms of $P_u(S)$.
\begin{align*}
P_v(S) &= \frac{2}{n \pi_n} \int_S \frac{dx}{\ell_\Sigma(x,v) \abs{x-v}^{n-1}}\\
&\ge \frac{\epsilon}{n \pi_n D_\Sigma} \int_S \frac{\abs{x-u} dx}{\ell_\Sigma(x,u) \abs{x-v}^{n}} & \text{\dots By \eqref{eq:lxvlxu}}\\
&\ge \frac{\epsilon}{n \pi_n D_\Sigma} \left( 1 + \frac{4}{n} \right)^{-n} \int_S \frac{dx}{\ell_\Sigma(x,u) \abs{x-u}^{n-1}} & \text{\dots By \eqref{eq:xvxu}}\\
&\ge \frac{\epsilon}{2 e^4 D_\Sigma} P_u(S) \;. 
\end{align*} 
Finally,
\begin{align*}
P_u(A) - P_v(A) &\le P_u(A) - P_v(S) \\
&\le P_u(A) - \frac{\epsilon}{2 e^4  D_\Sigma} P_u(S)\\ 
&\le P_u(A) - \frac{\epsilon}{2 e^4  D_\Sigma} \left( P_u(A) - \frac{3}{4} \right)\\ 
&= \frac{3\epsilon}{8 e^4  D_\Sigma} + \left( 1 - \frac{\epsilon}{2 e^4 D_\Sigma}  \right) P_u(A) \\
&\stackrel{(a)}{\le} \frac{3\epsilon}{8 e^4  D_\Sigma} + 1 - \frac{4\epsilon}{8 e^4 D_\Sigma} \\
&= 1 - \frac{\epsilon}{8 e^4  D_\Sigma} \;. 
\end{align*}
In the step (a), we used the fact that $D_\Sigma\ge \epsilon$ and $P_u(A)\leq 1$.
\end{proof}

\begin{proof}[Proof of Lemma~\ref{lem:s-cond}]
Let $\{S_1,S_2\}$ be a partitioning of $\Sigma$. Define 
\begin{align*}
\Sigma_1 &= \left\{x\in S_1\ :\  P_{x}(S_2) \le \delta \right\} \,,\\ 
\Sigma_2 &= \left\{x\in S_2\ :\  P_{x}(S_1) \le \delta \right\}\,, \\
\Sigma_3 &= \Sigma\setminus \Sigma_1\setminus \Sigma_2 \;.
\end{align*}
\textbf{Case 1:} $\vol(\Sigma_1) \le \vol(S_1)/2$. We have that
\begin{align*}
\int_{S_1} P_{x}(S_2) dx &\ge \int_{S_1\setminus \Sigma_1} P_{x}(S_2) dx \ge \delta \vol(S_1\setminus \Sigma_1) \ge \frac{\delta}{2} \vol(S_1) \;.
\end{align*}
Thus,
\begin{align*}
\frac{1}{\min\{\vol(S_1), \vol(S_2)\}}\int_{S_1} P_{x}(S_2) dx &\ge \frac{\delta}{2} \;.
\end{align*}
\textbf{Case 2:} $\vol(\Sigma_1) > \vol(S_1)/2$ and $\vol(\Sigma_2) > \vol(S_2)/2$. Similar to the argument in the previous case, 
\[
\int_{S_1} P_{x}(S_2) \ge \delta\, \vol(S_1\setminus \Sigma_1)\,,
\]
and
\[
\int_{S_1} P_{x}(S_2)  = \int_{S_2} P_{x}(S_1) \ge \delta\, \vol(S_2\setminus \Sigma_2)\;.
\]
Thus,
\[
\int_{S_1} P_{x}(S_2) \ge \frac{\delta}{2}\, \vol(\Sigma\setminus\Sigma_1\setminus\Sigma_2) =  \frac{\delta}{2}\, \vol(\Sigma_3)\;.
\]
Let $\Omega_i=g^{-1}(\Sigma_i)$ for $i=1,2,3$. Define
\[
(u(x),v(x)) = \argmin_{u\in\Omega_1,v\in\Omega_2, \{u,v,x\} \text{ are collinear}} d_\Omega(u,v)\,,\qquad h(x)=(1/3) \min(1, d_\Omega(u(x), v(x))) \;.
\]
By definition, $h(x)$ satisfies condition of Theorem~\ref{thm:iso-ineq}. Let $\epsilon = \frac{r}{2n}$ and notice that $\vol(\Omega^\epsilon) \ge \vol(\Omega)/2$. We have that
\begin{align*}
\int_{S_1} P_{x}(S_2) &\ge \frac{\delta}{2}\, \vol(\Omega_3) \\
&\ge \frac{\delta}{2} \bE_\Omega(h(x)) \min(\vol(\Omega_1), \vol(\Omega_2)) \\
&= \frac{\delta}{4} \bE_{\Omega^\epsilon}(h(x)) \min(\vol(\Sigma_1), \vol(\Sigma_2)) \;.
\end{align*}
Let $x\in\Omega^\epsilon$. We consider two cases. \underline{In the first case}, $\abs{u(x) - v(x)} \ge \epsilon/10$. Thus, 
\[
d_\Omega(u(x), v(x)) \ge \frac{4}{D_\Omega} \abs{u(x)-v(x)} \ge \frac{2}{5 n D_\Omega} \;.
\]
\underline{In the second case}, $\abs{u(x)-v(x)}< \epsilon/10$, then $\abs{u(x)-x}\le \epsilon/10$ and $\abs{v(x)-x}\le \epsilon/10$. Thus, $u,v\in\Omega^{\epsilon'}$ for $\epsilon'=9\epsilon/10$. Thus by Assumption~\ref{ass:mapping}, $g(u), g(v) \in \Sigma^{\epsilon''}$ for $\epsilon'' = 9\epsilon/(10 L_\Omega)$. By Lemma~\ref{lem:lipschitz}, 
\[
d_\Omega(u(x), v(x)) \ge \frac{\widetilde d_\Sigma(g(u(x)), g(v(x)))}{4 L_\Sigma^2 L_\Omega^2 R_{\epsilon'} (1+2 R_{\epsilon'})} \;.
\]
Next we lower bound $\widetilde d_\Sigma(g(u), g(v))$. If $g(u)$ and $g(v)$ cannot see each other, then $\widetilde d_\Sigma(g(u), g(v)) \ge 8\epsilon''/D_\Sigma$. Next we assume that $g(u)$ and $g(v)$ see each other. Because $g(u)\in \Sigma_1$ and $g(v)\in \Sigma_2$, 
\[
d_{tv} (P_{g(u)} - P_{g(v)}) \ge 1 - P_{g(u)}(S_2) - P_{g(v)}(S_1) \ge 1 - 2 \delta = 1 - \frac{\epsilon''}{8 e^4 D_\Sigma} \;.
\]
Lemma~\ref{lem:total-var}, applied to $g(u),g(v)\in\Sigma^{\epsilon''}$, gives us that 
\[
d_\Sigma(g(u), g(v)) \ge \frac{\epsilon''}{24 D_\Sigma} \qquad \mbox{or} \qquad \abs{g(u)-g(v)} \ge \frac{2}{\sqrt{n}} \min\left(\frac{2 F(g(u))}{\sqrt{n}}, G\epsilon'' \right) \;. 
\]
By \eqref{eq:Fu}, $F(g(u)) \ge \epsilon''/16$. We get the desired lower bound by taking a minimum over all cases. 
\end{proof}